\documentclass[letterpaper, 10 pt, conference]{ieeeconf}  

\IEEEoverridecommandlockouts                              

\usepackage{graphics} 
\usepackage{epsfig} 
\usepackage{mathptmx} 
\usepackage{times} 
\usepackage{amsmath} 
\usepackage{amssymb}  
\usepackage{graphicx}
\usepackage{caption}
\usepackage{subcaption}
\usepackage{cite}
\usepackage{xcolor}
\usepackage{tcolorbox}
\usepackage[ruled,vlined]{algorithm2e}
\usepackage[super]{nth}
\usepackage[utf8]{inputenc}
\usepackage{upgreek}

\usepackage{soul}
\setstcolor{red}

\makeatletter
\let\NAT@parse\undefined
\makeatother
\usepackage{hyperref}

\newtheorem{theorem}{Theorem}[section]

\title{\LARGE \bf
On the steady--state behavior of finite--control--set MPC\\ with an application to high--precision power amplifiers
}

\author{Duo Xu, Sander Damsma, Mircea Lazar  
\thanks{This work is funded by the EU Horizon 2020 research project IT2 (IC Technology for the 2nm Node), Grant agreement ID: 875999. 
}
\thanks{All authors are with the Control Systems Group, Department of Electrical Engineering, Eindhoven University of Technology, The Netherlands.}%
\thanks{E-mails: {\tt\small d.xu@tue.nl, sanderdamsma@gmail.com, m.lazar@tue.nl.} }
}

\begin{document}

\maketitle
\thispagestyle{empty}
\pagestyle{empty}
\begin{abstract}
Motivated by increasing precision requirements for switched power amplifiers, this paper addresses the problem of model predictive control (MPC) design for discrete--time linear systems with a finite control set (FCS). Typically, existing solutions for FCS--MPC penalize the output tracking error and the control input rate of change, which can lead to arbitrary switching among the available discrete control inputs and unpredictable steady--state behavior. To improve the steady--state behavior of FCS--MPC, in this paper we design a cost function that penalizes the tracking error with respect to a state and input steady--state limit cycle. We prove that if a suitable terminal cost is added to the FCS--MPC algorithm convergence to the limit cycle is ensured. The developed methodology is validated in direct switching control of a power amplifier for high--precision motion systems, where it significantly improves the steady--state output current ripple.    
\end{abstract}

\section{INTRODUCTION}
Presently, the advances in modern control theory and microprocessors open up the application of MPC in power electronics \cite{MPC_role}. Depending on whether a separate modulator is employed or not, MPC for power electronics can be classified into two main categories: indirect, averaged--model based MPC and direct switching control MPC. Similar to classical control schemes for power electronics \cite{Spinu_Constrianed}, indirect MPC computes a continuous duty cycle to drive the switches via a pulse width modulator. One of the drawbacks using averaging techniques is that the performance of the system states during each sampling interval is omitted.
Thus, direct control of the switches with MPC was proposed, which helps to evaluate the ripple and constraint performance in a more precise fashion. This approach is referred to as finite--control--set  MPC (FCS--MPC) and it is the favored and widely researched MPC method in academia \cite{rodriguez2012predictive}. 
In \cite{FCSMPC-coupled-inductor-buck-boost}, FCS--MPC is implemented on a coupled inductor buck boost converter, where it improves the current tracking performance and reduces the converter power losses.
In \cite{FCSMPC-non-minimum-phase-boost}, a boost converter is considered, and FCS--MPC is shown to eliminate non-minimum phase behavior for output voltage control.
In \cite{FCSMPC-PMSM1,FCS-MPC-PMSM}, the FCS--MPC is implemented for a  Permanent Magnet Synchronous Machine, to enable optimal switching among various operation strategies.

Typically, the FCS-MPC problem is formulated as an online integer programming optimization problem, which grows in complexity significantly with the system dimension/number of switches and with the length of the prediction horizon.
A growing body of research in this field has been dedicated to address the numerical complexity aspects of FCS--MPC, see for example, \cite{7499836}, \cite{9104994} and the references therein.
\cite{FCSMPC-high-speed-FCSMPC} achieves the FPGA implementation based on the approximate dynamic programming, which enables the sampling times below 25 $\upmu$s.

Another relevant research direction for FCS--MPC is the cost function design so that steady--state behavior is improved. FCS--MPC mainly utilizes a cost function that penalizes the output tracking error and the discrete input rate of change \cite{Guidelines_FCS_MPC}. This avoids computation of a steady--state solution for the switched affine system, but it also results in unpredictable steady--state behavior, which is undesirable in high--precision applications \cite{9180048}. For such applications, it is desirable to control the steady--state solution of a power amplifier so that the output current ripple is minimized. Steady--state issues for FCS--MPC were originally considered in \cite{6256724}, which proposed intermediate sampling or integral error terms to improve steady--state behavior.

Recently, in \cite{limit_cycle} a method was proposed to design and compute optimal steady--state limit cycles for switched affine systems. 
To deal with the switched dynamics or, equivalently, with a finite control set, therein the steady--state solution is designed as a limit cycle that stays close to the desired output reference. 
Motivated by this recent idea and the current steady--state issues of FCS--MPC, in this paper we propose a time--varying cost function that penalizes the tracking error with respect to pre--computed state and input limit cycle trajectories. We prove that if a suitable terminal penalty is added and the open loop dynamics is stable, the developed FCS--MPC algorithm converges to the pre--computed limit cycle asymptotically. Simulation results for an industrial high--precision power amplifier show that the designed FCS--MPC algorithm can achieve a significantly smaller  steady--state output current ripple compared with the existing output tracking FCS--MPC design.

\section{PRELIMINARIES} \label{section2}

The sets of real and integer numbers are denoted as $\mathbb{R}$ and $\mathbb{Z}$ respectively. $\mathbb{Z}_{[1,n]}$ denotes the set of $n$ first natural numbers. $\{0,1\}$ denotes the binary set.
The modulo operator is defined as $c$ = $a$ mod $b$ where $c$ is the remainder of the Euclidean division between the integers $a$ and $b$.

The model of switched power amplifiers is captured by the following class of discrete--time linear systems with a finite control set:
\begin{equation} \label{equ:linear_sys}
\begin{split}
x(k+1) &= A x(k) + B u(k),
\\
y(k) &= C x(k),
\end{split}
\end{equation}
where $x(k) \in \mathbb{R}^{n}$ is the system state, i.e., capacitor voltages and inductor currents, $u(k) \in \{0,1\}^{m}$ denotes the discrete system inputs, which represent the switching state and $y(k) \in \mathbb{R}^{q}$ is the system output. Thus, the vector $u(k)$ belongs to a finite control set $\mathbb{U}$, and the number of its elements is $2^m$.

\subsection{Standard FCS-MPC for output tracking}
In order to define the finite control set MPC optimization problem, let $x_{0|k}=x(k)$ denote the system state at time $k$ and let $x_{i|k}$ ($y_{i|k}$) and $u_{i|k}$ denote the predicted state (output) and input at time $k+i$ over a prediction horizon $N$. The sequence of the predicted inputs $U_{k}$ can be expressed as
\begin{equation}
U_{k}=\left\{u_{0|k}, \ldots, u_{N-1|k}\right\}.
\end{equation}
To track an output reference $y_{ref}$, the cost function of the standard FCS-MPC is defined as \cite{Guidelines_FCS_MPC}
\begin{equation} \label{equ:standard_FCS_MPC_cost}
J\left(x(k), U_{k}, y_{ref}\right)=F_{T}\left(y_{N|k},y_{ref}\right)+\sum_{i=0}^{N-1} l\left(y_{i|k}, u_{i|k}, y_{ref}\right)
\end{equation}
where $F_{T}(\cdot)$ is the terminal cost and $l(\cdot)$ is the stage cost. These cost functions can be formulated as follows
\begin{equation}\label{equ:terminal_cost}
F_{T}\left(y_{N|k}, y_{ref}\right)=\left(y_{N|k}-y_{ref}\right)^\top P\left(y_{N|k}-y_{ref}\right),
\end{equation}
\begin{equation} \label{equ:stage_cost}
\begin{split}
l\left(y_{i|k}, u_{i|k}, y_{ref}\right)=&\left(y_{i|k}-y_{ref}\right)^\top Q \left(y_{i|k}-y_{ref}\right)\\
&+ \Delta u_{i|k}^\top R \Delta u_{i|k}.
\end{split}
\end{equation}
In \eqref{equ:stage_cost}, $\Delta u_{i|k} :=u_{i|k}-u_{i-1|k}$ for all $i=0,\ldots,N-1$ and $u_{-1|k}:=u_{0|k-1}$. The penalty of the difference between two consecutive inputs $\Delta u_{i|k}$ will result in less switching.
However, this also affects the steady--state behavior of the output and system states, as it will be illustrated in this paper.

Based on the cost function \eqref{equ:standard_FCS_MPC_cost}, the system model \eqref{equ:linear_sys} and (optional) state constraints ($x(k) \in \mathbb{X}$) or output constraints ($y(k) \in \mathbb{Y}$), the tracking FCS-MPC can be formulated as the following integer optimization problem, assuming that the full state is measurable (or observable) \cite{Guidelines_FCS_MPC}:
\begin{equation} \label{equ:standard_FCS_MPC_optmization}
\begin{aligned}
\min _{u_{0|k},\ldots,u_{N-1|k} } &  J\left(x(k), U_{k}, y_{ref}\right) \\
\text { s.t. } 
& x_{i+1|k}=Ax_{i|k}+Bu_{i|k}, \quad &\forall i=0, \ldots, N-1,\\
& y_{i|k}=Cx_{i|k}, \quad &\forall i=0, \ldots, N,\\
& u_{i|k} \in \{0,1\}^{m}, \quad &\forall i=0, \ldots, N-1, \\
& x_{i|k} \in \mathbb{X}\,\, (\text{or } y_{i|k} \in \mathbb{Y}), \quad &\forall i=1, \ldots, N.\\
\end{aligned}
\end{equation}
One common solution to this integer optimization problem \eqref{equ:standard_FCS_MPC_optmization} is to use the brute-force approach of exhaustive enumeration. Alternatively, 
dedicated optimization algorithms such as branch-and-bound and non-trivial prediction horizon formulations \cite{non_trivial_prediction} can be considered. The optimal sequence of the predicted input set $U_{k}^*=\left\{u_{0|k}^{*}, \ldots, u_{N-1|k}^{*}\right\}$ can be derived after solving the optimization problem \eqref{equ:standard_FCS_MPC_optmization}, the first element $u_{0|k}^*$ in $U_{k}^*$ is selected as the optimal input.
Notice that because the set of optimization variables is finite, a (possibly non--unique) global optimum exists.

\subsection{Problem formulation}
For a switched affine system such as \eqref{equ:linear_sys}, the system output $y(k)$ cannot be controlled to maintain the exact output reference $y_{ref}$ for the averaged model at steady state; instead, the output will oscillate around the averaged model reference output, which will generate a periodic trajectory, also known as a limit cycle. 

In this paper, the first considered problem is the design of a limit cycle $\overline{X}_c:=\{\overline{x}(0),\ldots,\overline{x}(p-1)\}$ at steady state with a fixed length $p\geq 1$. Such a predefined limit cycle $\overline{X}$ is optimized to generate a suitable trajectory of $y(k)$, while satisfying certain criteria such as least mean deviation with respect to $y_{ref}$ or minimum ripple amplitude.

The second considered problem is how to design the cost function of FCS--MPC such that the convergence to the pre--computed limit cycle is achieved.
\section{Main results} \label{section3}

\subsection{Computation of steady--state limit cycles}
For system \eqref{equ:linear_sys}, any combination of a $p$-periodic input sequence ${U}_c=\{{u}_{c}(0) ,\ldots , {u}_{c}(p-1)\}$ can generate a corresponding $p$-periodic limit cycle $\overline{X}_c$. Here, we use the index $c$ to denote one of possible combination of discrete inputs, which are $(2^m)^p$ in total for a limit cycle of length $p$, i.e., $c \in \mathbb{Z}_{[1,(2^m)^p]}$.

Based on the property of a $p$-periodic limit cycle, i.e., ${x}_{c}(p)={x}_{c}(0)$, we obtain via \eqref{equ:linear_sys} that
\begin{equation}
{x}_{c}(0) = (I-A^p)^{-1} 
\begin{bmatrix}
A^{p-1}B & A^{p-2}B & ... & B
\end{bmatrix}
\begin{bmatrix}
{u}_{c}(0)\\{u}_{c}(1)\\...\\{u}_{c}(p-1)
\end{bmatrix},
\end{equation}
\begin{equation}
{x}_{c}(i+1)=A{x}_{c}(i)+B{u}_{c}(i), \quad \forall i=0, \ldots, p-2.
\end{equation}

Thus, a family of limit cycles $\mathfrak{X}$ can be derived as follows
\begin{equation}
\mathfrak{X}=\left\{{X}_c:  c \in \mathbb{Z}_{[1,(2^m)^p]}\right\},
\end{equation}
where $\mathfrak{X}$ contains $(2^m)^p$ different limit cycles.

According to \cite{limit_cycle}, an optimal limit cycle $\overline{X}_{c}$ which satisfies certain criteria can be calculated by solving an optimization problem, i.e.,
\begin{equation} \label{equ:limit_cycle_optimize}
\overline{X}_{c}=\arg \min_{{X}_{c} \in \mathfrak{X}} \frac{1}{p} \sum_{n=0}^{p-1}\left\|\Gamma\left(C{x}_{c}(n)-y_{ref}\right)\right\|,
\end{equation}
where $\Gamma$ is a tuning parameter and different norms can be used. 
The optimization problem \eqref{equ:limit_cycle_optimize} aims to find the limit cycle with the least mean distance between $\Gamma C\overline{x}_{c}(n)$ and $\Gamma y_{ref}$.

\subsection{Formulation of FCS--MPC for limit cycle tracking}

Different from tracking a constant output reference $y_{ref}$ as formulated in \eqref{equ:standard_FCS_MPC_cost}, the cost function of tracking the optimal limit cycle  $\overline{X}_{c}$ and its corresponding optimal input $\overline{U}_{c}$ is defined as
\begin{equation} \label{equ:Limit_Cycle_FCS_MPC_cost}
\begin{split}
J\left(x(k), U_{k}, \overline{X}_{c}, \overline{U}_{c}\right)=&F_{T}\left(x_{N|k},\overline{x}_{N|k}\right)\\
&+\sum_{i=0}^{N-1} l\left(x_{i|k}, u_{i|k}, \overline{x}_{i|k}, \overline{u}_{i|k}\right)
\end{split}
\end{equation}
where
\begin{equation} \label{equ:mod}
\begin{split}
\overline{x}_{i|k} = \overline{x}_{c}(k+i\; \text{mod}\; p), \quad &\forall i=0, \ldots, N,\\
\overline{u}_{i|k} = \overline{u}_{c}(k+i\; \text{mod}\; p), \quad &\forall i=0, \ldots, N-1,
\end{split}
\end{equation}
\begin{equation} \label{equ:terminal_cost}
F_{T}\left(x_{N|k},\overline{x}_{N|k}\right)=\left(x_{N|k}-\overline{x}_{N|k}\right)^\top P\left(x_{N|k}-\overline{x}_{N|k}\right),
\end{equation}
\begin{equation}
\begin{split}
l\left(x_{i|k}, u_{i|k}, \overline{x}_{i|k}, \overline{u}_{i|k}\right)&=
\left(x_{i|k}-\overline{x}_{i|k}\right)^\top Q \left(x_{i|k}-\overline{x}_{i|k}\right)\\
&+ \left(u_{i|k}-\overline{u}_{i|k}\right)^\top R \left(u_{i|k}-\overline{u}_{i|k}\right).
\end{split}
\end{equation}

Therefore, the FCS--MPC for tracking an optimal limit cycle can be reformulated as the integer optimization problem:
\begin{equation} \label{equ:Limit_Cycle_FCS_MPC_optmization}
\begin{aligned}
\min _{u_{0|k},\ldots,u_{N-1|k} } &  J\left(x(k), U_{k}, \overline{X}_{c},\overline{U}_{c}\right) \\
\text { s.t. } 
& x_{i+1|k}=Ax_{i|k}+Bu_{i|k}, \quad &\forall i=0, \ldots, N-1\\
& u_{i|k} \in \{0,1\}^{m}, \quad &\forall i=0, \ldots, N-1,
\\
& x_{i|k} \in \mathbb{X}, \quad &\forall i=1, \ldots, N.\\
\end{aligned}
\end{equation}
The optimal input signal $u^*(k)$ is the first element $u_{0|k}^*$ in $U_{k}^*$.

\subsection{Convergence analysis}
Next, we provide a convergence result for FCS--MPC with limit cycle tracking based on the shifted sequence of inputs approach, as typically done in stabilizing MPC with a continuous (infinite) control set. 
\begin{theorem}
Consider the limit cycle tracking problem using FCS-MPC for the system \eqref{equ:linear_sys}.
If the optimization problem \eqref{equ:Limit_Cycle_FCS_MPC_optmization} is always feasible,
the asymptotic convergence to the specified optimal limit cycle is guaranteed
if the matrix $A$ is Schur stable and the terminal cost matrix $P$ in \eqref{equ:terminal_cost} is chosen a positive definite solution of the linear matrix inequality $-P+Q+A^\top PA\prec0$.
\end{theorem}
\begin{proof}
Consider the optimal input sequence at time $k$
\begin{equation}
U_{k}^*=\left\{u_{0|k}^{*}, \ldots, u_{N-1|k}^{*}\right\}.
\end{equation}
A shifted input sequence $\widetilde{U}_{k+1}$ for the next time instance $k+1$ is defined as
\begin{equation} \label{equ:U_shift}
\begin{split}
\widetilde{U}_{k+1}&=\left\{\widetilde{u}_{0|k+1}, \ldots, \widetilde{u}_{N-2|k+1}, \widetilde{u}_{N-1|k+1}\right\}\\
               &=\left\{u_{1|k}^{*}, \ldots, u_{N-1|k}^{*}, \widetilde{u}_{N-1|k+1}\right\},
\end{split}
\end{equation}
where the last term $\widetilde{u}_{N-1|k+1}$ in \eqref{equ:U_shift} is not assigned yet.

Assuming that the optimization problem \eqref{equ:Limit_Cycle_FCS_MPC_optmization} is always feasible (i.e., there exists at least one sequence in the set of discrete inputs that results in a state trajectory that does not violate the state constraints), the cost function at time instance $k+1$ associated with the shifted input sequence satisfies the following relation:
\begin{equation} \label{equ:dV}
\begin{split}
J \big(x(k+1), \widetilde{U}_{k+1},& \overline{X}_{c}, \overline{U}_{c}\big)\\
= &J\left(x(k), U_{k}^*, \overline{X}_{c},\overline{U}_{c}\right)\\
- &l\left(x_{0|k}, u_{0|k}^*, \overline{x}_{0|k}, \overline{u}_{0|k}\right)\\
- &F_{T}\left(x_{N|k}^*,\overline{x}_{N|k}\right)\\
+ &l\left(x_{N-1|k+1}, \widetilde{u}_{N-1|k+1}, \overline{x}_{N-1|k+1}, \overline{u}_{N-1|k+1}\right)\\
+ &F_{T}\left(x_{N|k+1},\overline{x}_{N|k+1}\right).
\end{split}
\end{equation}
As done in standard stabilizing MPC, by optimality at time $k+1$, in order to ensure that the optimal cost function $J$ is strictly decreasing at time $k+1$, it is necessary that the sum of the last three terms in \eqref{equ:dV} is non--positive. Then 
\begin{equation}
\begin{split}
&J\left(x(k+1), U_{k+1}^*, \overline{X}_{c^*},\overline{U}_{c^*}\right)
- J\left(x(k), U_{k}^*, \overline{X}_{c^*},\overline{U}_{c^*}\right)\\
\leq &J\left(x(k+1), \widetilde{U}_{k+1}, \overline{X}_{c^*},\overline{U}_{c^*}\right)
- J\left(x(k), U_{k}^*, \overline{X}_{c^*},\overline{U}_{c^*}\right)\\
\leq &-l\left(x_{0|k}, u_{0|k}^*, \overline{x}_{0|k}, \overline{u}_{0|k}\right) < 0, \quad \forall x(k)\neq \overline{x}_{0|k}.
\end{split}
\end{equation}
The required condition on the last three terms in \eqref{equ:dV} is 
\begin{equation} \label{equ:last_3_term}
\begin{split}
-&\left(x_{N|k}^*-\overline{x}_{N|k}\right)^\top P\left(x_{N|k}^*-\overline{x}_{N|k}\right)\\
+&\left(x_{N-1|k+1}-\overline{x}_{N-1|k+1}\right)^\top Q \left(x_{N-1|k+1}-\overline{x}_{N-1|k+1}\right)\\
+&\left(\widetilde{u}_{N-1|k+1}-\overline{u}_{N-1|k+1}\right)^\top R \left(\widetilde{u}_{N-1|k+1}-\overline{u}_{N-1|k+1}\right)\\
+&\left(x_{N|k+1}-\overline{x}_{N|k+1}\right)^\top P \left(x_{N|k+1}-\overline{x}_{N|k+1}\right)\leq 0.
\end{split}
\end{equation}

Based on the choice of the shifted sequence, we have that $x_{N-1|k+1}=x_{N|k}^*$, $x_{N|k+1}=Ax_{N-1|k+1}+B\widetilde{u}_{N-1|k+1}$. The periodic reference  $\overline{x}_{i|k}$ defined in \eqref{equ:mod} gives that $\overline{x}_{N-1|k+1}=\overline{x}_{N|k}$, $\overline{x}_{N|k+1}=A\overline{x}_{N-1|k+1}+B\overline{u}_{N-1|k+1}$. Then, by assigning the last term $\widetilde{u}_{N-1|k+1}$ in \eqref{equ:U_shift} equal to $\overline{u}_{N-1|k+1}$, \eqref{equ:last_3_term} becomes
\begin{equation}
\left(x_{N|k}^*-\overline{x}_{N|k}\right)^\top (-P+Q+A^\top PA) \left(x_{N|k}^*-\overline{x}_{N|k}\right).\\
\end{equation}
This is true under the assumption that $x_{N|k+1}=Ax_{N|k}^\ast+B\bar{u}_{N-1|k+1}\in\mathbb{X}$, or if there are no state constraints.

Thus, it can be concluded that, if the open loop dynamics $A$ of system \eqref{equ:linear_sys} is strictly stable, then for any prediction horizon $N$, there always exists a terminal cost matrix $P\succ 0$ such that $-P+Q+A^\top PA\prec0$ and the optimal cost function $J$ is strictly decreasing at every discrete--time instant $k\geq 1$. Since $J$ is positive definite and lower bounded by zero, we have that 
\begin{equation} \label{equ:limit_J}
    \begin{split}
        \lim_{k\to\infty} J\big(x(k), U_{k}^\ast,& \overline{X}_c, \overline{U}_c\big)\\
        =&\lim_{k\to\infty} 
        \left(x_{N|k}^\ast-\overline{x}_{N|k}\right)^\top P\left(x_{N|k}^\ast-\overline{x}_{N|k}\right)\\
        +&\sum_{i=0}^{N-1}\left(
        \left(x_{i|k}^\ast-\overline{x}_{i|k}\right)^\top Q \left(x_{i|k}^\ast-\overline{x}_{i|k}\right)\right. \\ +&\left.
        \left(u_{i|k}^\ast-\overline{u}_{i|k}\right)^\top R \left(u_{i|k}^\ast-\overline{u}_{i|k}\right)\right)\\
        =&0.
    \end{split}
\end{equation}
Since each term in \eqref{equ:limit_J} is a positive definite quadratic form for $P,Q,R\succ 0$ it follows that
\begin{equation} \label{equ:limit_J}
\begin{split}
\lim_{k\to\infty} \|x(k)-\overline{x}_{0|k}\| = 0, \quad \lim_{k\to\infty} \|u^*(k)-\overline{u}_{0|k}\| = 0,
\end{split}
\end{equation}
where $\overline{x}_{0|k}$ and $\overline{u}_{0|k}$ are periodically time--varying within their corresponding steady--state limit cycle trajectories as defined in \eqref{equ:mod}. By the shift rule for converging sequences it further follows that for $\forall i=1,\ldots,p-1$, we have
\begin{equation} \label{equ:limit_Jshift}
\lim_{k\to\infty} \|x(k+i)-\overline{x}_{0|k+i}\| = 0, \quad
\lim_{k\to\infty} \|u^*(k+i)-\overline{u}_{0|k+i}\| = 0.
\end{equation}
Thus, it follows that for all $k\geq 0$, the closed--loop state and input trajectories $x(k)$ and $u^*(k)$ asymptotically converge to $\overline{x}_{0|k} = \overline{x}_{c}(k\; \text{mod}\; p)$ and $\overline{u}_{0|k} = \overline{u}_{c}(k\; \text{mod}\; p)$, respectively. 
\end{proof}
The condition for calculating the terminal cost matrix $P$ here is similar to the stability conditions for the standard linear MPC;
the difference is that the convergence is guaranteed with respect to a periodic varying reference (a limit cycle).
Typically, the stability of FCS--MPC is analyzed using the quantized local controllers and practical stability \cite{practical_stability} notions. So this provides a useful result for FCS--MPC.

\section{APPLICATION TO POWER AMPLIFIERS}\label{section4}

In this section we demonstrate the standard FCS--MPC and the developed FCS--MPC design framework, respectively, for controlling an industrial power amplifier used in high--precision motion systems for the lithography industry.

\subsection{Industrial Power Amplifier Topology and Modelling}
\begin{figure}[h] 
\centering
\includegraphics[width=0.3\textwidth]{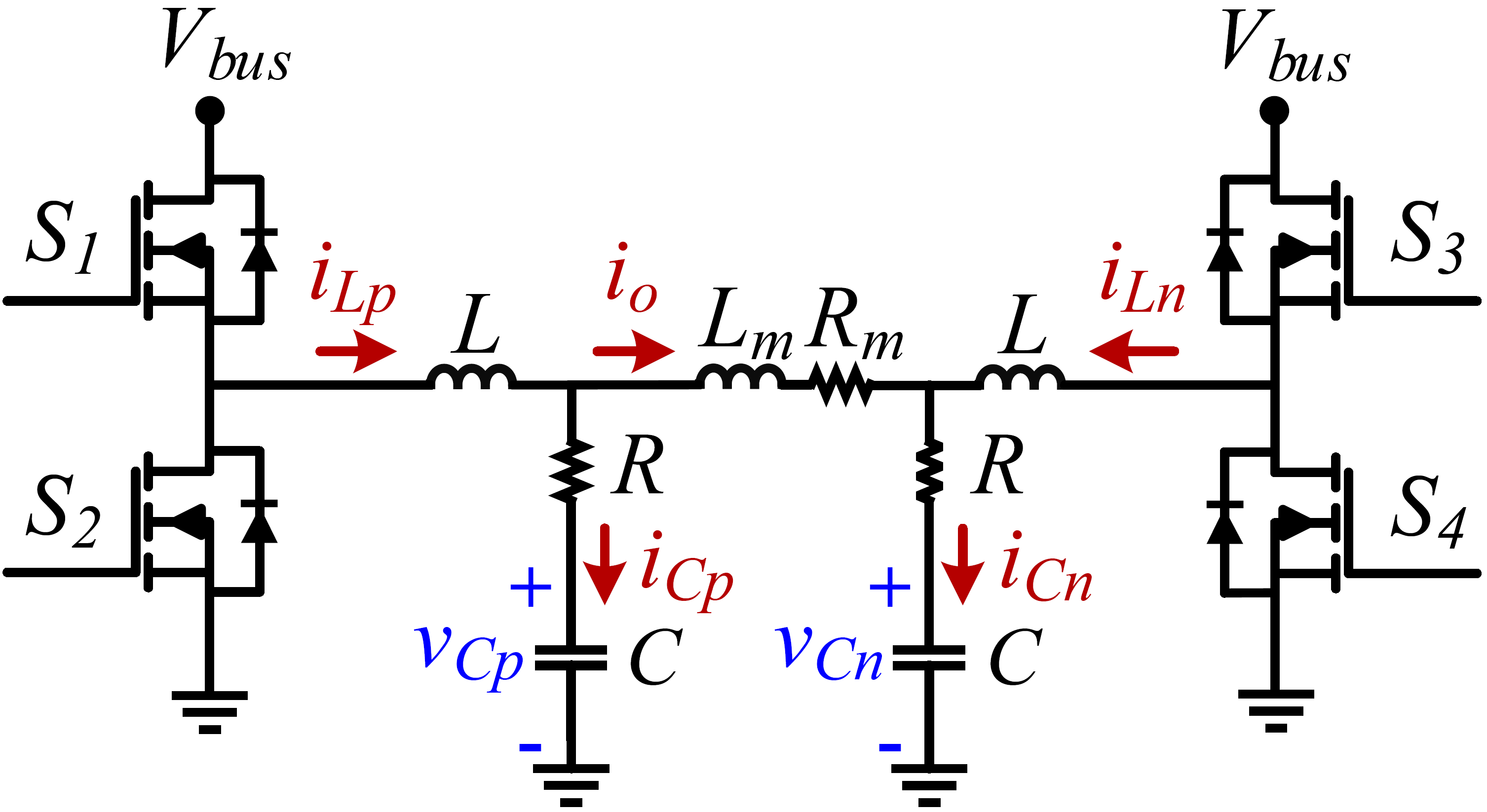}
\caption{Schematic of the industrial power amplifier.}
\label{fig:PADC_schematic}
\end{figure}
\begin{table}[h] 
\centering
\caption{Industrial Power Amplifier Circuit Parameters}
\label{Table:Circuit Parameters}
\begin{tabular}{llll}
\hline
parameter & Symbol &  Value \\ \hline
Bus voltage & $V_{bus}$ &  $360V$ \\
Power stage inductance & $L$ &  $44\mu H$ \\
Power stage capacitance & $C$ &  $0.4\mu F$ \\
Power stage parasitic resistance & $R$ &  $62.2 \mu \Omega$ \\
Load model inductance & $L_m$ &  $20mH$ \\
Load model resistance & $R_m$ &  $10\Omega$ \\
\hline
\end{tabular}
\end{table}
The schematic of the considered high precision industrial current amplifier is depicted in Fig. \ref{fig:PADC_schematic} \cite{Faran_PADC}. The amplifier is divided into two identical power stages connected to the load model. These power stages will be identified as the positive and negative power stages, with subscript $p$ or $n$. 
The circuit parameters are listed in Table~\ref{Table:Circuit Parameters}.
The power stages are switched synchronously, as defined in \eqref{eq:input_switching}. The switching mode of the amplifier can be expressed as a input vector $u(t)$ as in \eqref{eq:input_vector}. This power amplifier has four operation modes in total with respect to the switching states ($S_1,S_2,S_3,S_4$) as listed in Table \ref{tabel:operation_modes}. This can alternatively be expressed as a finite control set of states $\mathbb{S}$ in \eqref{eq:inputspace}.
\begin{equation} \label{eq:input_switching}
    S_p = S_1 = \overline{S_2}, \quad S_n = S_3 = \overline{S_4}
\end{equation}
\begin{equation} \label{eq:input_vector}
    u(t) = 
    \begin{bmatrix}
        S_p(t)\\
        S_n(t)
    \end{bmatrix},\;
    \begin{matrix}
    S_p \in \{0,1\} \\
    S_n \in \{0,1\}
    \end{matrix}
\end{equation}
\begin{equation} \label{eq:inputspace}
   u(t) \in \mathbb{S} := \left\{
    \begin{bmatrix}
        0\\
        0
    \end{bmatrix},
    \begin{bmatrix}
        1\\
        0
    \end{bmatrix},
    \begin{bmatrix}
        0\\
        1
    \end{bmatrix},
    \begin{bmatrix}
        1\\
        1
    \end{bmatrix}
    \right\}.
\end{equation}
\begin{table}[!h]
	\centering
	\caption{Operation Modes of the Power Amplifier}
	\label{tabel:operation_modes}
    \begin{tabular}{|c|c|c|c|c|}
    \hline Modes & $S_{1}$ & $S_{2}$ & $S_{3}$ & $S_{4}$ \\
    \hline 1 & open & close & open & close \\
    \hline 2 & open & close & close & open \\
    \hline 3 & close & open & open & close \\
    \hline 4 & close & open & close & open \\
    \hline
\end{tabular}
\end{table}
The switched continuous state-space model \eqref{equ:switched_model} can be derived by using Kirchhoff's circuit laws for each of the switch modes. The system states are $x=[\begin{matrix} i_{Lp} & v_{Cp} & i_{Ln} & v_{Cn} & i_{o} \end{matrix}]^\top$. 
The measured output is the output current $i_{o}$.
The obtained model is the following: 
\begin{equation} \label{equ:switched_model}
\begin{split}
\dot{x}(t)&= A_c x(t) + B_c u(t),
\\
y(t)&=C_c x(t),
\end{split}
\end{equation}
where
\begin{equation}
\begin{gathered}
A_{c} =\left[\begin{smallmatrix}
-\frac{R}{L} & -\frac{1}{L} & 0 & 0 & \frac{R}{L} \\
\frac{1}{C} & 0 & 0 & 0 & -\frac{1}{C} \\
0 & 0 & -\frac{R}{L} & -\frac{1}{L} & -\frac{R}{L} \\
0 & 0 & \frac{1}{C} & 0 & \frac{1}{C} \\
\frac{R}{L_{m}} & \frac{1}{L_{m}} & -\frac{R}{L_{m}} & -\frac{1}{L_{m}} & -\frac{2 R+R_{m}}{L_{m}}
\end{smallmatrix}\right]
B_c = \left[\begin{smallmatrix}
        \frac{V_{bus}}{L} &0   \\
        0&0\\
        0&\frac{V_{bus}}{L}\\
        0&0\\
        0&0
    \end{smallmatrix}\right]
C_c=\left[\begin{smallmatrix}
        0 \\ 0 \\ 0 \\ 0 \\ 1
    \end{smallmatrix}\right]^\top_.
    \label{Eq: State space C}
\end{gathered}
\end{equation}
Above, $A_c \in \mathbb{R}^{5\times5}$, $B_c \in \mathbb{R}^{5\times2}$ and $C_c \in \mathbb{R}^{1\times5}$. 
The linear switched model is discretized using the zero-order-hold method with the sampling frequency $f_s$, i.e.,
\begin{equation} \label{eq:averaged_model_dis}
\begin{gathered}
x(k+1)=Ax(k)+Bu(k),
\\
y(k)= Cx(k),
\end{gathered}
\end{equation}
where $x=[\begin{matrix} i_{Lp} & v_{Cp} & i_{Ln} & v_{Cn} & i_{o} \end{matrix}]^\top$, $y=i_{o}$, $u(k) \in \mathbb{S}$, $A \in \mathbb{R}^{5x5}$, $B \in \mathbb{R}^{5x2}$ and $C \in \mathbb{R}^{1x5}$.

\subsection{Simulation setup}
The performances of the standard FCS-MPC and limit cycle tracking FCS-MPC are simulated using MATLAB. The integer optimization problems \eqref{equ:standard_FCS_MPC_optmization}\eqref{equ:Limit_Cycle_FCS_MPC_optmization} are solved with the state--of--the--art integer optimization Gurobi solver. All simulations are performed with the following specifications: Desktop computer, AMD Ryzen 7 3800XT CPU 4.0GHz, 32GB RAM, 64--bit OS. The system output $y(k)$ is controlled to track a step signal of the output reference $y_{ref}=6$A. The system state $x(k)$ is sampled with a frequency of $f_s = 400$ kHz. The FCS-MPC controller is updated with the same sampling frequency $f_s$.
The initial state is $x(0)=\begin{bmatrix}
 0  &  0  &  0 & 0 & 0
\end{bmatrix}^\top$, as it typically happens at start--up.

The optimal limit cycle for such output reference $y_{ref}$ is calculated using \eqref{equ:limit_cycle_optimize} with 2--norm and with $\Gamma = 1$, which minimizes the mean distance between $\Gamma C\overline{x}_{c}(n)$ and $\Gamma y_{ref}$. The periodicity $p$ is set to $p=6$, which is the smallest length that can attain an averaged current of $i_o$ around 6A. The calculated optimal input and state limit cycles are as follows:
\begin{equation} \label{equ:Uss}
 \overline{U}_c=\left\{
    \begin{bmatrix} 1\\ 0\end{bmatrix}, 
    \begin{bmatrix} 0\\ 1\end{bmatrix}, 
    \begin{bmatrix} 1\\ 0\end{bmatrix}, 
    \begin{bmatrix} 0\\ 0\end{bmatrix}, 
    \begin{bmatrix} 0\\ 0\end{bmatrix}, 
    \begin{bmatrix} 0\\ 0\end{bmatrix}
    \right\},
\end{equation}
\begin{equation} \label{equ:Xss}
 \overline{X}_{c}=
    \left\{\begin{smallmatrix}
    -7.3138	& 8.2415 & 3.7586 & 19.3138	& 11.0141 & 0.9858\\
    112.2068 & 76.5487 & 76.5490 & 112.2073 & 171.2443 & 171.2440\\
    -14.0433 & -16.2847 & 4.2848 & 2.0433 & -2.9709 & -9.0291\\
    67.7927 & 8.7557 & 8.7560 & 67.7932 & 103.4513 & 103.4510\\
    6.0003 & 5.9987 & 6.0013 & 5.9997 & 5.9994 & 6.0006\\
    \end{smallmatrix}\right\}.
\end{equation}
In order to have an intuitive illustration, the input switching state can be transformed to the corresponding operation mode according to Table \ref{tabel:operation_modes}.
Thus, the optimal input limit cycle \eqref{equ:Uss} is transformed into a sequence of operation modes, which is $\{3,2,3,1,1,1\}$.

\subsection{Standard FCS-MPC performance}
As illustrated in Section~\ref{section2}.A, the standard FCS-MPC controller is formulated as in Table \ref{Table:simulations_FCS_MPC}. The state constraints are omitted, thus the performance of the unconstrained standard FCS-MPC is evaluated for 2 different prediction horizons. The reason for selecting $N=3$, 4 is that the results do not converge when $N<3$ and the performances behave the same when $N\geq4$.
\begin{table}[h!] \caption{Standard FCS-MPC Control Setup}
\centering
\label{Table:simulations_FCS_MPC}
\begin{tabular}{llll}
\hline
Parameter & Symbol &  Value \\ \hline
Prediction horizon & $N$ &  3, 4 \\
Weighting matrix & $P=Q$ & 1 \\
Weighting matrix & $R$ & diag($10^{-4}$, $10^{-4}$)\\
State constraints & $\mathbb{X}$ &  $x_{i|k} \in \mathbb{R}^{5}$ \\
\hline
\end{tabular}
\end{table}

The performance of the output current $i_o$ with standard FCS--MPC is shown in Fig. \ref{fig:Simulation1_1}, the $k$ in the $x$--axis denotes the simulation steps. 
Besides, the comparison of current ripple and average computation time for different prediction horizons is listed in Table \ref{Table:Sim1}.
For both $N=3$ and $N=4$, the output current $i_o$ has a fast transient and good steady--state performance.
The standard FCS--MPC achieves a current ripple of 18.9068 mA and 17.8828 mA in steady state after an overshoot of 14 mA and 33 mA respectively.
Compared to the optimal current ripple, which is $2.6153$ mA calculated from \eqref{equ:Xss}, the standard FCS--MPC generates a significantly larger current ripple.
Fig. \ref{fig:Simulation1_2} depicts the control action of the operation mode at steady--state. Both prediction horizons result in the same repeating sequence of $\{3,1,1,1,1,1\}$.
\begin{figure}[h!] 
\centering
\includegraphics[width=0.47\textwidth]{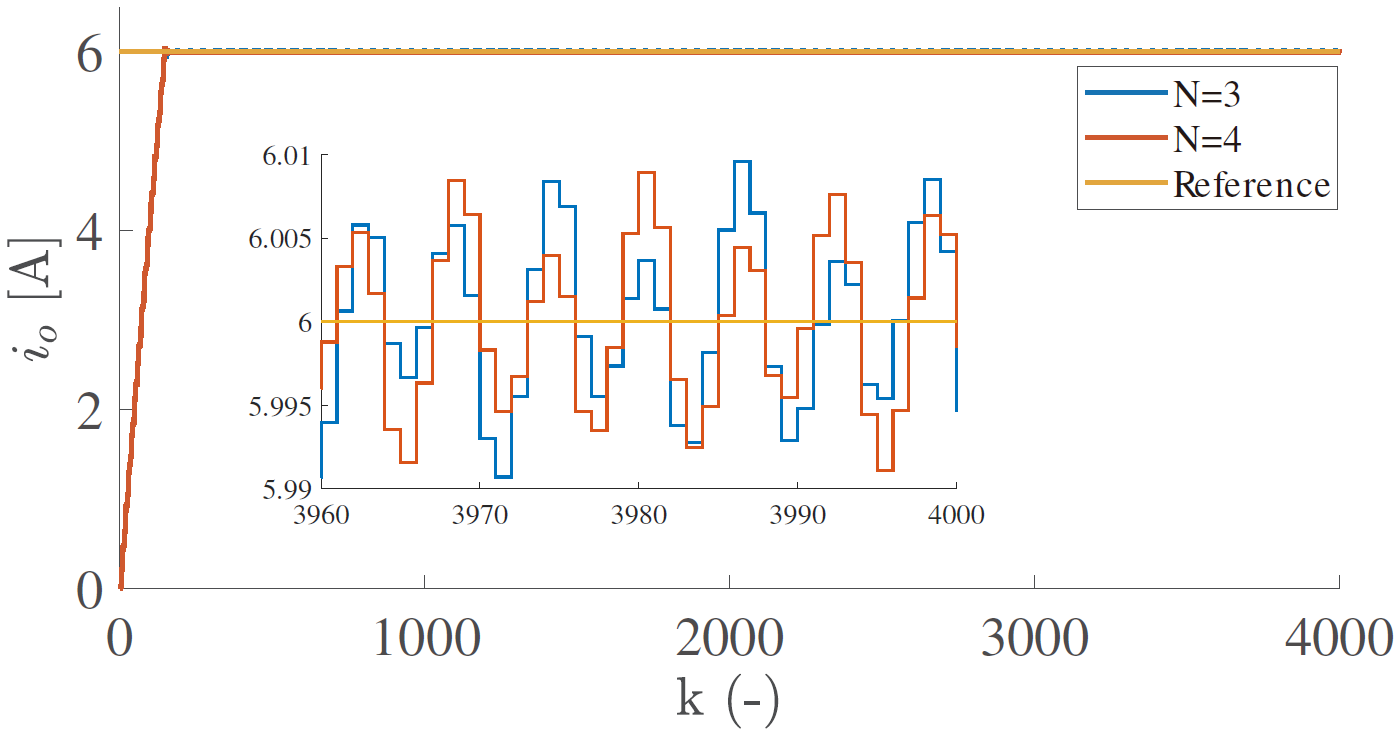}
\caption{Standard FCS--MPC: output current $i_o$.}
\label{fig:Simulation1_1}
\end{figure}
\begin{table}[h!] \caption{Standard FCS-MPC Performance}
\centering
\label{Table:Sim1}
\begin{tabular}{llll}
\hline
Prediction horizon  & Ripple of $i_o$  &  Average computation time  \\ \hline
$N=3$ & 18.9068 [mA] &  4.65 [ms] \\
$N=4$ & 17.8828 [mA] & 10.95 [ms] \\
\hline
\end{tabular}
\end{table}
\begin{figure}[h!] 
\centering
\includegraphics[width=0.47\textwidth]{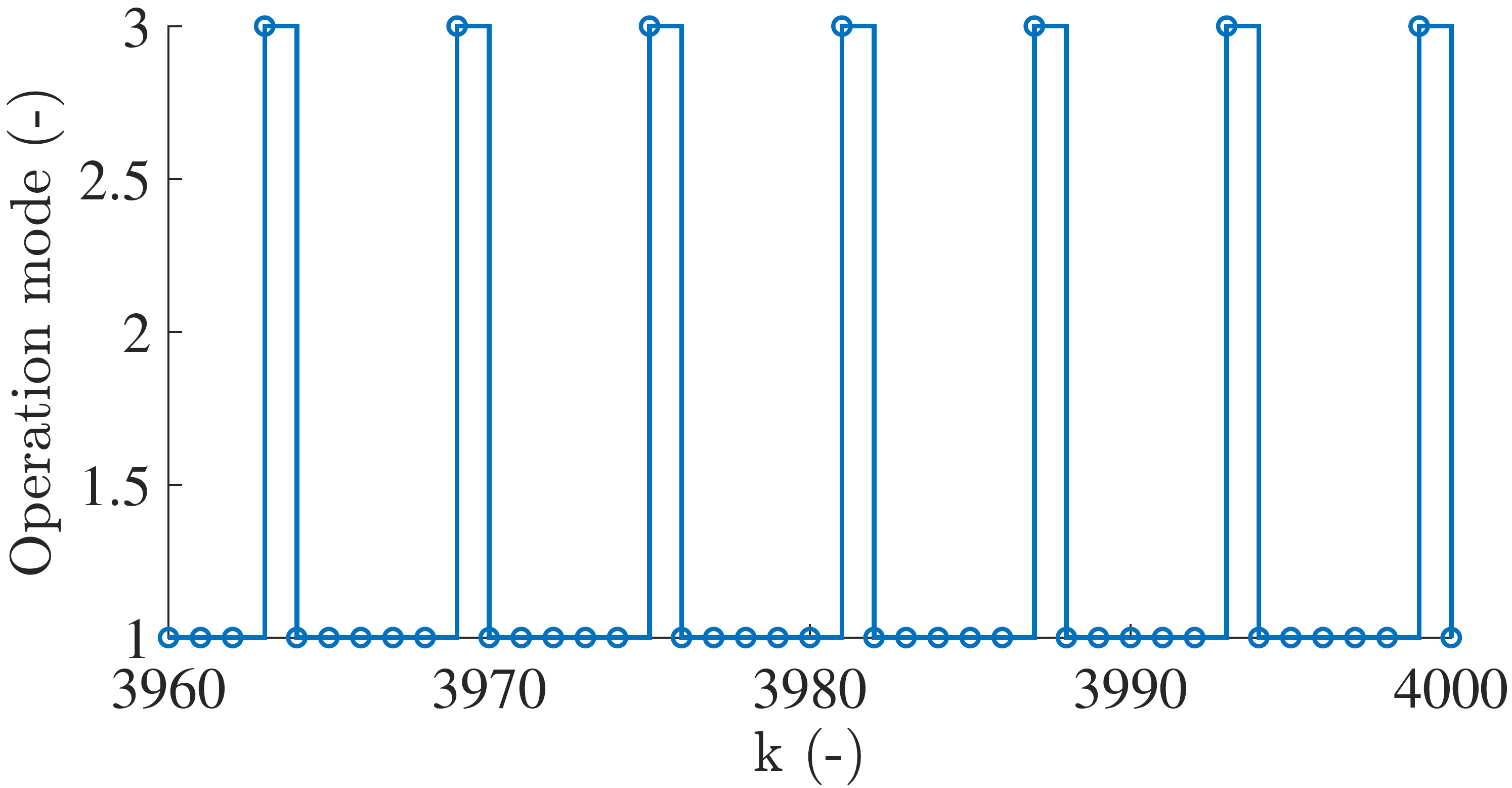}
\caption{Standard FCS--MPC: operation mode.}
\label{fig:Simulation1_2}
\end{figure}

\subsection{Limit cycle tracking FCS-MPC performance}
The limit cycle tracking FCS-MPC is implemented as illustrated in Section~\ref{section3}.B. 
According to \textbf{Theorem III.1}, the diagonal structure of $P$ was imposed and a feasible solution was found with the YALMIP MATLAB toolbox \cite{Lofberg2004}, but in general a diagonal $P$ may not exist. Table \ref{Table:simulations_steady_state_FCS_MPC} lists the detailed control setup. The last element in $Q$ is the penalty for output current $i_o$, and it is set to 1 which makes it comparable with the standard FCS--MPC.
\begin{table}[h!] \caption{Limit Cycle Tracking FCS-MPC Control Setup}
\centering
\label{Table:simulations_steady_state_FCS_MPC}
\begin{tabular}{llll}
\hline
Parameter & Symbol &  Value \\ \hline
Prediction horizon & $N$ &  4, 6, 8 \\
Weighting matrix & $P$ &  diag($2\cdot10^{4},189,2\cdot10^{4},189,9.5\cdot10^{6}$) \\
Weighting matrix & $Q$ &  diag($L/L_m$, $C/L_m$, $L/L_m$, $C/L_m$, 1)\\
Weighting matrix & $R$ &  diag($5\cdot10^{-2}$, $5\cdot10^{-2}$) \\
State constraints & $\mathbb{X}$ &  $x_{i|k} \in \mathbb{R}^{5}$ \\
\hline
\end{tabular}
\end{table}

The performance of output current $i_o$ with limit cycle tracking FCS--MPC is depicted in Fig. \ref{fig:Simulation2_1}. 
In addition, the comparison of current ripple and average computation time for different prediction horizons is listed in Table \ref{Table:Sim2}.
It is observed that a higher prediction horizon gives a faster transient response and better steady--state performance. For $N=8$, at steady--state, it results in the least offset to the optimal limit cycle and produces a peak to peak current ripple of 4.2102 mA which is close to the optimal current ripple of 2.6153 mA.
Similarly, Fig. \ref{fig:Simulation2_2} shows the control action of the operation mode, which behaves the same as the optimal input limit cycle.
Fig. \ref{fig:Simulation2_4} illustrates the state trajectory of the output current $i_{o}$ and the inductor current $i_{Lp}$ at steady--state when $N=8$. The state trajectory converges to the optimal state limit cycle $\overline{X}_{c}$ asymptotically.

\begin{figure}[h!] 
\centering
\includegraphics[width=0.47\textwidth]{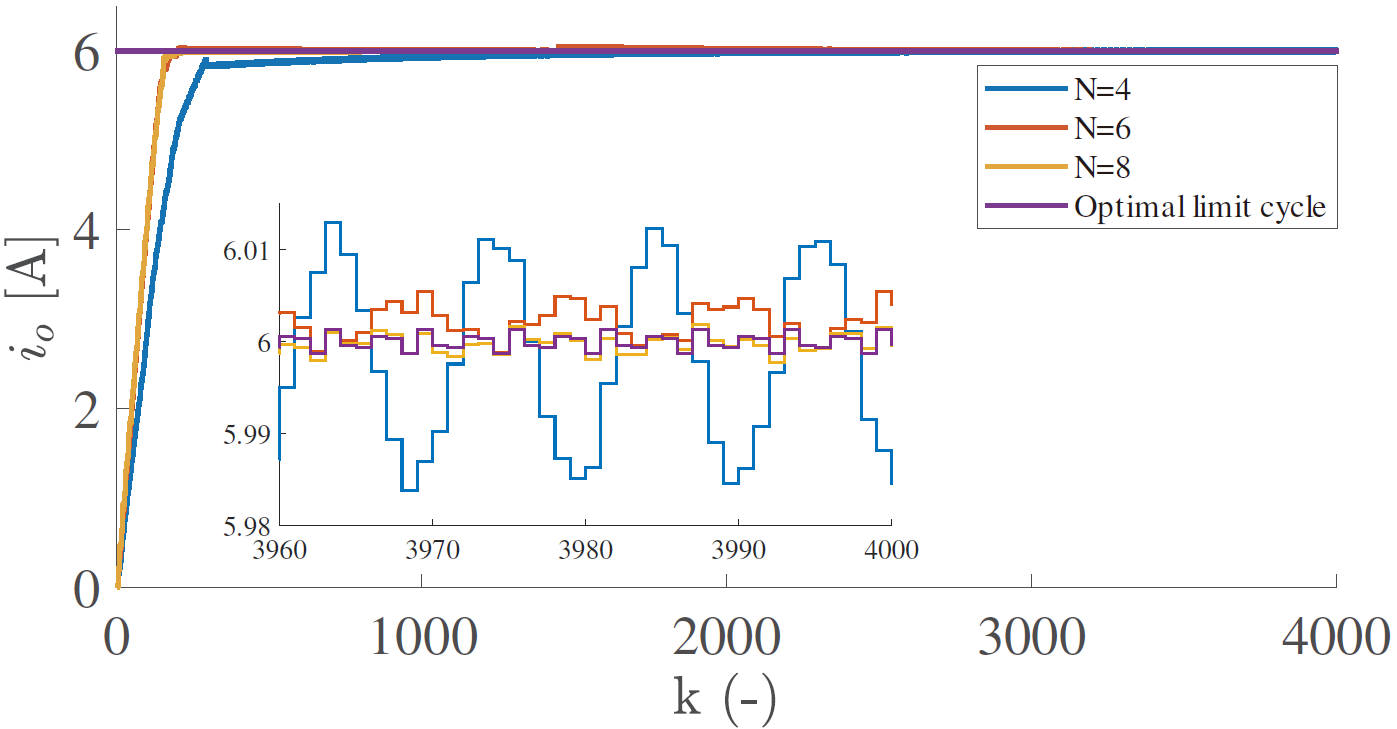}
\caption{Limit cycle tracking FCS--MPC: output current $i_o$.}
\label{fig:Simulation2_1}
\end{figure}
\begin{table}[h!] \caption{Limit Cycle Tracking FCS-MPC Performance}
\centering
\label{Table:Sim2}
\begin{tabular}{llll}
\hline
Prediction horizon  & Ripple of $i_o$  &  Average computation time  \\ \hline
$N=4$ & 29.1691 [mA] & 4.91 [ms] \\
$N=6$ & 6.8609 [mA] & 11.95 [ms] \\
$N=8$ & 4.2102 [mA] & 23.14 [ms] \\
\hline
\end{tabular}
\end{table}
\begin{figure}[h!] 
\centering
\includegraphics[width=0.47\textwidth]{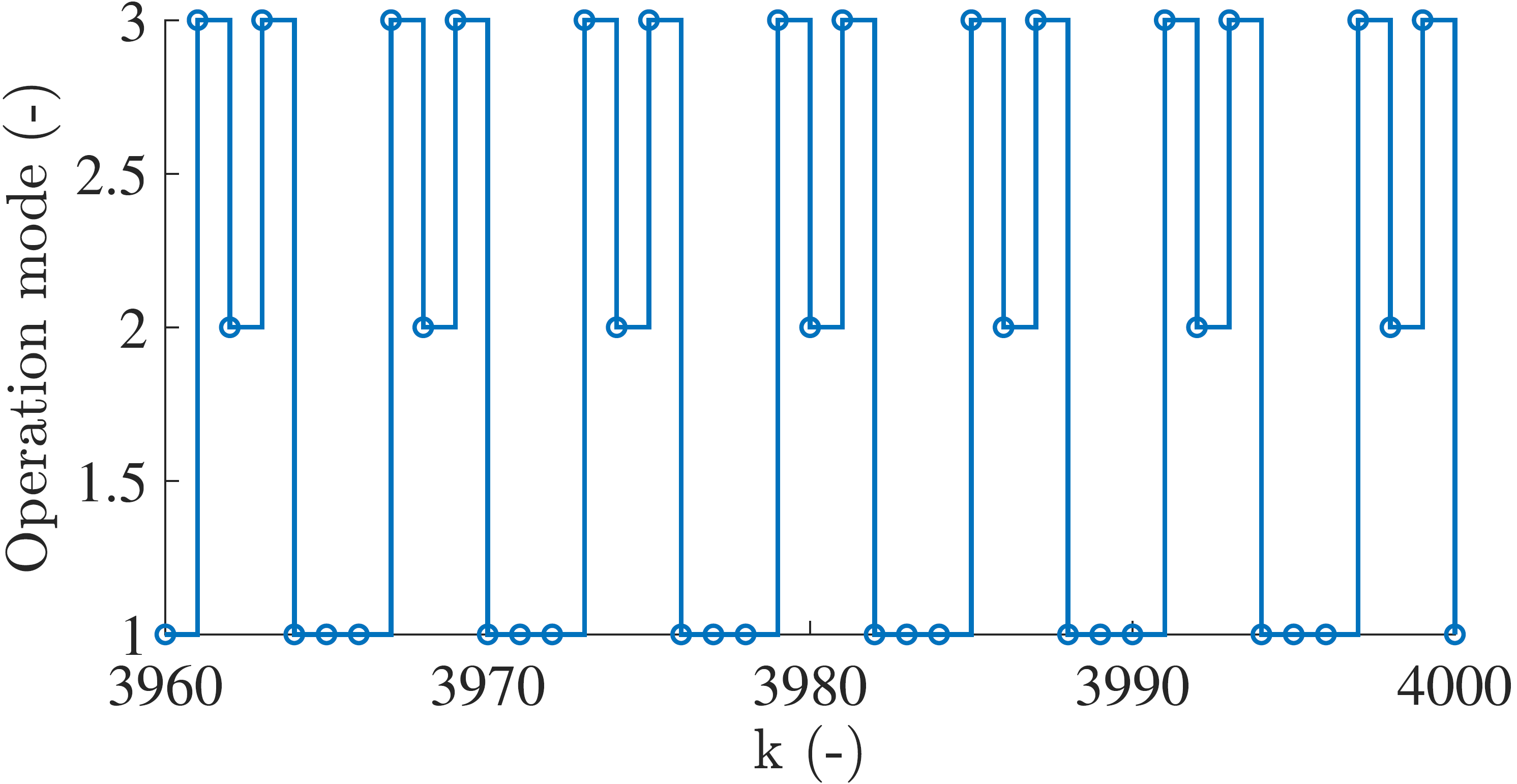}
\caption{Limit cycle tracking FCS--MPC: operation mode.}
\label{fig:Simulation2_2}
\end{figure}
\begin{figure}[h!] 
\centering
\includegraphics[width=0.47\textwidth]{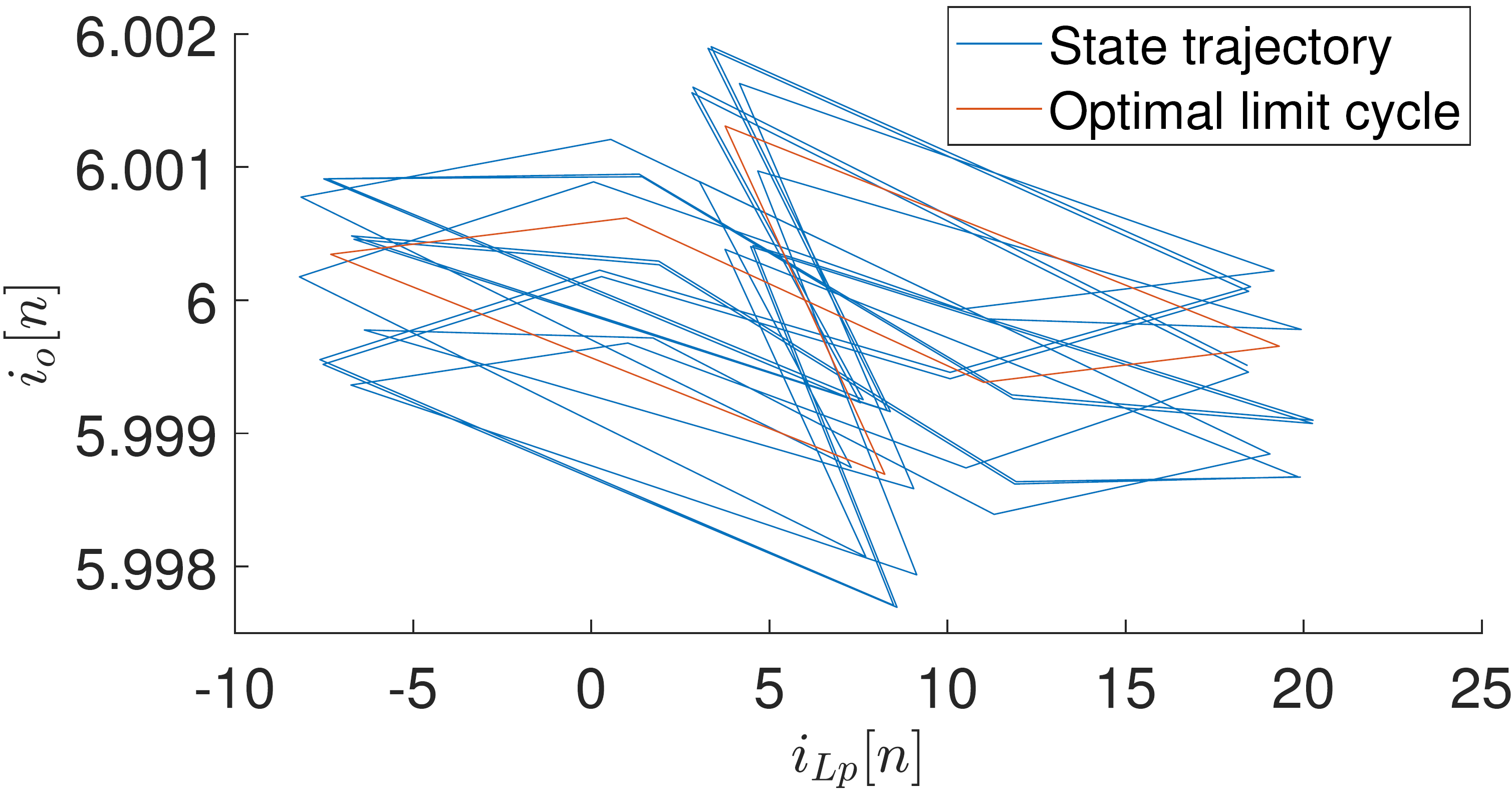}
\caption{Limit cycle tracking FCS--MPC: state trajectory.}
\label{fig:Simulation2_4}
\end{figure}

\section{CONCLUSIONS}\label{section5}

Inspired by the computation of an optimal steady--state limit cycle for a linear system with finite control inputs set, this paper proposed an FCS--MPC design to realize the tracking control of a limit cycle. The cost function of the FCS-MPC penalizes the tracking error with respect to the pre--computed state and input limit cycle trajectories. Then, the convergence to the pre--computed limit cycle was proven for a suitable terminal penalty stable open--loop dynamics. Compared with the standard FCS--MPC, the proposed limit cycle tracking FCS--MPC enables a better steady--state performance. The efficiency of the developed FCS--MPC design was demonstrated in output current control of a high--precision power amplifier, where it resulted in improved output current ripple.

\bibliographystyle{IEEEtran}
	\bibliography{SectionBib}

\addtolength{\textheight}{-12cm}   
 




\end{document}